\documentclass[12pt,oneside]{amsart}

\usepackage{amsthm,amsmath,amssymb}
\usepackage{mathrsfs}
\usepackage{enumerate}
\usepackage{amsfonts}
\usepackage{verbatim}
\usepackage{amsbsy}
\usepackage{amsmath}
\usepackage{amssymb}
\usepackage{MnSymbol}
\usepackage{mathrsfs}
\usepackage{xcolor}
\usepackage{cite}

\newtheorem{theorem}{Theorem}[section]

\newtheorem{proposition}[theorem]{Proposition}

\newtheorem{corollary}[theorem]{Corollary} 
\theoremstyle{definition}
\newtheorem{definition}[theorem]{Definition}

\theoremstyle{remark}
\newtheorem{remark}[theorem]{Remark}
\newtheorem{example}[theorem]{Example}

\DeclareMathOperator{\res}{res}

\DeclareMathOperator{\Sym}{Sym}

\newcommand*{\permcomb}[4][0mu]{{{}^{#3}\mkern#1#2_{#4}}}
\newcommand*{\perm}[1][-3mu]{\permcomb[#1]{P}}

\begin{document}

\title[Compositions of Minimal Reaction Systems   ]{Compositions of Functions and Permutations Specified by Minimal Reaction Systems}
\author{Wen Chean Teh}
\address{School of Mathematical Sciences\\
Universiti Sains Malaysia\\
11800 USM, Malaysia}
\email{dasmenteh@usm.my}

\begin{abstract}
This paper studies mathematical properties of reaction systems that was introduced by Enrenfeucht and Rozenberg as computational models inspired by biochemical reaction in the living cells. In particular, we continue the study on the generative power of functions specified by minimal reaction systems under composition initiated by Salomaa. Allowing degenerate reaction systems, 
functions specified by minimal reaction systems over a quarternary alphabet that are permutations generate the alternating group on the power set of the background set.
\end{abstract}

\maketitle
\section{Introduction}

Reaction systems, introduced by Ehrenfeucht and Rozenberg \cite{ehrenfeucht2007reaction}, is a natural computing approach ultimately aiming to understand the functioning of the living cells.
Studies on reaction systems have been on various diverse lines, on the basic framework, as well as on various extensions. For motivational surveys on reaction systems, we refer the reader to \cite{ehrenfeucht2012qualitative, ehrenfeucht2012reaction}.

This work belongs to the line of research that focus on the mathematical properties of functions specified by reaction systems (for example, \cite{ehrenfeucht2011functions, manzoni2014simple, salomaa2015applications, salomaa2012functions, teh2017irreducible, teh2018irreducible}), exclusively those specified by minimal reaction systems \cite{ehrenfeucht2012minimal, salomaa2013minimal, salomaa2014compositions, salomaa2014minimal, salomaa2015two, salomaa2017minimal, azimi2017steady, teh2017minimal}. In \cite{salomaa2014compositions} Salomaa showed that the set of functions specified by nondegenerate (see Definition~\ref{2606a}) minimal reaction systems over a ternary alphabet each of which permute the nonempty proper subsets of the background set is closed under composition. This implies that nondegenerate minimal reaction systems have limited generative power because not every function specified by a reaction system can be generated from them. This paper revisits Salamaa's results while allowing degenerate reaction systems and further extends them to the quarternary alphabet.

 Section~2 provides the basic terminology for reaction systems and describes our representation of functions that are specified by reaction systems used in our computer simulation. The subsequent section revisits some of Salomaa's work in \cite{salomaa2014compositions} from different perspective and further extends his results to the quarternary alphabet. Our final section shows that functions specified by minimal reaction systems exhibit different mathematical behaviour when degeneracy is allowed. In particular, for the ternary alphabet,  three functions specified by degenerate minimal reaction systems possess full generative power under composition.

\section{Preliminaries}

\subsection*{Basic Notions of Reaction Systems}\label{1203a}

If $S$ is an arbitrary finite set, then the cardinality of $S$ is denoted by $\vert S\vert$ and the power set of $S$ is denoted by $2^S$.

\begin{definition}
	Suppose $S$ is a finite nonempty set. A \emph{reaction in $S$} is  a triple
	$a=(R_a,I_a, P_a)$, where $R_a$ and $I_a$ are disjoint (possibly empty) subsets of $S$ and $P_a$ is a nonempty subset of $S$. The sets $R_a$, $I_a$, and $P_a$ are the \emph{reactant set}, \emph{inhibitor set}, and \emph{product set} respectively.
\end{definition}

\begin{definition}\label{2606a}
	A \emph{reaction system (over $S$)} is a pair $\mathcal{A}=(S,A)$ where $S$ is a finite nonempty \emph{background set} and $A$ is a (possibly empty) set of reactions in $S$. We say that $\mathcal{A}$ is \emph{nondegenerate} if{f} $R_a$ and $I_a$ are both nonempty for every $a\in A$.
\end{definition}

From now onwards, $S$ is a fixed finite nonempty background set.

\begin{definition}
	Suppose \mbox{$\mathcal{A}=(S,A)$} is a reaction system. The function 
	$\res_{\mathcal{A}} \colon 2^S\rightarrow 2^S$  is defined by
	$$\res_{\mathcal{A}}(X)= \bigcup_{\substack{a\in A\\
			R_a\subseteq X, I_a\cap X=\emptyset} } P_a \quad,\quad \text{for all }X\subseteq S.   $$
\end{definition}

We may identify $\res_A$ with $\res_{\mathcal{A}}$ when $S$ is understood.

\begin{definition}
	Every function $f\colon 2^S\rightarrow 2^S$ is called an \emph{rs function over $S$}.
It is \emph{nondegenerate} if{f} $f(\emptyset)=f(S)=\emptyset$. We say that $f$ can be \emph{specified by a reaction system} $\mathcal{A}$ over $S$ if{f} $f=\res_{\mathcal{A}}$.
\end{definition}

In \cite{ehrenfeucht2012minimal} nondegeneracy of reaction systems is an adopted assumption. Hence, over there 
every rs function specified by a reaction system over $S$ is nondegenerate. However, 
unless stated explicitly otherwise, a reaction system can be degenerate in this work. Since every rs function over $S$ can be canonically specified by a unique maximally inhibited\footnote{A reaction system $(S,A)$ is \emph{maximally inhibited} if{f} 
	$I_a= S\backslash R_a$ for all $a\in A$.}
	reaction system over $S$, it follows that the class of rs functions over $S$ is exactly the class of functions specified by reaction systems over $S$.

\begin{definition}\cite{ehrenfeucht2012minimal}
	Suppose $f$ is an rs function over $S$. 
	\begin{enumerate}
		\item $f$ is \emph{union-subadditive} if{f} $f(X\cup Y)\subseteq f(X)\cup f(Y)$ for all $X,Y\subseteq S$.
		\item $f$ is \emph{intersection-subadditive} if{f} $f(X\cap Y)\subseteq f(X)\cup f(Y)$ for all $X,Y\subseteq S$.
	\end{enumerate}
\end{definition}

\begin{definition}\cite{ehrenfeucht2012minimal, teh2017minimal} \label{1507d}
	Suppose $\mathcal{A}=(S,A)$ is reaction system.
Then  $\mathcal{A}$ is \emph{minimal} if{f} $\vert R_a\vert\leq 1$
		and $\vert I_a\vert\leq 1$ for every reaction $a\in A$.
	\end{definition}

\begin{theorem} \cite{ehrenfeucht2012minimal, teh2017minimal}    \label{2305b}
	Suppose $f$ is an rs function over $S$. 
Then $f=\res_\mathcal{A}$ for some (possibly degenerate) minimal reaction system  $\mathcal{A}$ if and only if $f$ is both union-subadditive and intersection-subadditive.
\end{theorem}

The above characterization was obtained
originally in \cite{ehrenfeucht2012minimal} for the context of nondegenerate reaction systems. Later it was shown to remain valid even when the nondegeneracy assumption is dropped  \cite{teh2017minimal}. This partly motivates us to extend Salomaa's work in \cite{salomaa2014compositions} to account also for degenerate reaction systems.

\begin{definition}
Let $\mathcal{M}(S)$ denote the set of of rs functions over $S$ such that each can be specified by a (possibly degenerate) minimal reaction systems.	
\end{definition}

\subsection*{Representation of RS Functions}

Suppose $S=\{s_0, s_1, s_2, \dotsc, s_{n-1}\}$ is a finite set of size $n$. We can represent each subset of $S$ by a nonnegative integer less than $2^n$ through binary representation. Precisely,
the subset $\{ s_{i_0}, s_{i_1}, s_{i_2},\dotsc, s_{i_m}  \}$, where
$0\leq i_0<i_1<i_2< \dotsb < i_m \leq n-1$ is represented by
$2^{i_m}+\dotsb + 2^{i_2}+2^{i_1}+2^{i_0}$. The representating integer is bolded to make a distinction. 

Particularly, for the quarternary alphabet
$S= \{s_0, s_1, s_2, s_3\}$, the representation of its subsets are as follows.
\begin{align*}
&\emptyset   &   &\mathbf{0} &&\{s_2\}   &   &\mathbf{4}  &&\{s_3\}   &   &\mathbf{8} && \{s_2,s_3\}   &   &\mathbf{12}   \\
&\{s_0\}   &   &\mathbf{1} && \{s_0,s_2\}   &   &\mathbf{5} & &\{s_0,s_3\}   &   &\mathbf{9} && \{s_0,s_2,s_3\}   &   &\mathbf{13}   \\
&\{s_1\}   &   &\mathbf{2} && \{s_1,s_2\}   &   &\mathbf{6}&  &\{s_1,s_3\}   &   &\mathbf{10} && \{s_1,s_2,s_3\}   &   &\mathbf{14}         \\
&\{s_0,s_1\}   &   &\mathbf{3} && \{s_0,s_1,s_2\}   &   &\mathbf{7} & &\{s_0,s_1,s_3\}   &   &\mathbf{11} && S   &   &\mathbf{15}       
\end{align*}

An rs function over $S$ can thus be represented by a row vector of length $2^{\vert S\vert}$, where the $(i+1)$-th entry represents the image of $\mathbf{i}$ for $0\leq i\leq 2^{\vert S\vert}-1$.
For example, the rs function $f$ over $S=\{s_0,s_1,s_2\}$, where 
\begin{align*}
&	f(\mathbf{0} )={}\mathbf{4} , & &f(\mathbf{1} )=\mathbf{1}, & &f(\mathbf{2} )={}\mathbf{5}, && f(\mathbf{3})={}\mathbf{2},\\
&	f(\mathbf{4})={}\mathbf{7}, & &f(\mathbf{5})={}\mathbf{2}, && f(\mathbf{6})={}\mathbf{4},  && f(\mathbf{7})=\mathbf{6},
\end{align*}
can be represented by the row vector $[
\mathbf{4}\;\, \mathbf{1}\;\,\mathbf{5}\;\,\mathbf{2}\;\,\mathbf{7}\;\, \mathbf{2}\;\,\mathbf{4}\;\,\mathbf{6}]$. If $f$ permutes $2^S$, we may denote it using the usual cycle decomposition for permutations. For example, the permutation $[
\mathbf{0}\;\, \mathbf{1}\;\,\mathbf{5}\;\,\mathbf{3}\;\,\mathbf{7}\;\, \mathbf{2}\;\,\mathbf{4}\;\,\mathbf{6}]$ can be denoted by $(\mathbf{2}\;\,\mathbf{5})(
\mathbf{4}\;\, \mathbf{7}\;\,\mathbf{6})$.

Finally, set operations and relations are conveniently extended to these representations. For example, $\mathbf{3}\cap \mathbf{6}=\mathbf{2}$ as $\{s_0,s_1\}\cap \{s_1,s_2\}=\{s_1\}$
and $\mathbf{3}\subseteq \mathbf{7}$ as 
$\{s_0,s_1\}\subseteq \{s_0,s_1,s_2\}$.

\section{Extension of Salomaa's Results}\label{2306a}

To avoid trivialities, we further assume that $\vert S\vert \geq 3$.

\begin{definition}
	Suppose $\sigma \colon S\rightarrow S$. Let $f_\sigma$ be the (nondegenerate) rs function over $S$ defined by 
	\begin{itemize}
		\item $f_\sigma(\emptyset)=f_\sigma(S)=\emptyset$;
		\item $f_\sigma(X)=\sigma[X]$ for all $\emptyset \neq X\subsetneq S$,
	\end{itemize}
where $\sigma[X]=\{\,\sigma(x)\mid x\in X\,    \} $.
	Let $f_{\sigma}^c(X)=f_\sigma (S\backslash X)$ for all $X\subseteq S$.
\end{definition}	

Let $\Sym(A)$ denote the symmetric group on a set $A$.

\begin{definition}
	\begin{enumerate}
		\item $\mathcal{F}_U(S)=\{\, f_\sigma\mid \sigma\colon S\rightarrow S            \,\}$.
		\item $\mathcal{F}_U^P(S)=\{\, f_\sigma\mid \sigma\in \Sym(S)            \,\} \cup  \{\, f_\sigma^c\mid \sigma\in \Sym(S)        \,\}          $.	
	\end{enumerate}	
\end{definition}

The class $\mathcal{F}_U(S)$ was introduced in \cite{salomaa2014compositions} and shown to be closed under composition. Furthermore, every rs function in $\mathcal{F}_U(S)$ belongs to $\mathcal{M}(S)$.
The subscript $U$ is due to the fact that
$f_{\sigma}(X)=\bigcup_{x\in X} f_\sigma(\{x\})$
for every nonempty proper subset $X$ of $S$.
Meanwhile, the superscript $P$ in $\mathcal{F}_U^P(S)$
refers to permutation.

\begin{proposition}\label{1806b}
	The class $\mathcal{F}_U^P(S)$ is closed under composition.	
\end{proposition}

\begin{proof}
	Suppose $\sigma, \tau\in \Sym(S)$. We will only show that $f_{\sigma}^c \circ f_\tau= f_{\sigma\circ \tau}^c$. Clearly, $(f_{\sigma}^c \circ f_\tau)(\emptyset)=
	(f_{\sigma}^c \circ f_\tau)(S)=\emptyset= f_{\sigma\circ \tau}^c(\emptyset)= f_{\sigma\circ \tau}^c(S)$.
	Suppose $X$ is a nonempty proper subset of $S$.
	Then 
	$$(f_{\sigma}^c \circ f_\tau)(X)= f_{\sigma}^c (\tau[X]  ) = f_{\sigma} ( S\backslash \tau[X]     ).$$
	Since $\tau\in \Sym(S)$, it follows that $S\backslash \tau[X]= \tau[ S\backslash X]$. Hence,
	$$ (f_{\sigma}^c \circ f_\tau)(X)	
=  \sigma[	\tau[ S\backslash X] ]
	=(\sigma\circ\tau)[S\backslash X]=f_{\sigma\circ \tau}^c(X).$$   
	Similarly, it can be shown that $f_{\sigma} \circ f_{\tau}^c= f_{\sigma\circ \tau}^c$ 
	and $f_{\sigma} \circ f_\tau= f_{\sigma}^c \circ f_{\tau}^c=    f_{\sigma\circ \tau}$.
\end{proof}

\begin{proposition}
	Every rs function in the class $\mathcal{F}_U^P(S)$ 
permutes the nonempty proper subsets of $S$	and
	belongs to $\mathcal{M}(S)$.
\end{proposition}

\begin{proof}
Suppose $\sigma\in \Sym(S)$. Since $\sigma$ is  a permutation, it follows that $f_\sigma$   permutes the nonempty proper subsets;
thus so is $f^c_\sigma$ from the definition.
Also, since $f_\sigma\in \mathcal{F}_U(S)$, it follows that $f\in \mathcal{M}(S)$
and thus $f^c_\sigma  \in \mathcal{M}(S) $ by Remark~\ref{1606a}. 
(Alternatively, it can verified directly that $f_{\sigma}$ and $f^c_\sigma$ are union-subadditive and intersection-subadditive.)	
\end{proof}

The following result was found in \cite{salomaa2014compositions} by exhaustively going through all the 720 rs functions that
permute the nonempty proper subsets of $S$ for $\vert S\vert=3$. Here, we give a formal proof and reformulate the result in terms of $\mathcal{F}_U^P(S)$.

\begin{theorem}\cite{salomaa2014compositions}\label{0106a}
	For $\vert S\vert=3$, the class $\mathcal{F}_U^P(S)$ is exactly the set of nondegenerate rs functions over $S$ that
	permute the nonempty proper subsets of $S$ and belong to $\mathcal{M}(S)$.
\end{theorem}

\begin{proof}
	Suppose $f$ is nondegenerate, permutes the nonempty proper subsets of $S$, and belongs to $\mathcal{M}(S)$. By Theorem~\ref{2305b}, $f$ is union-subadditive and intersection-subadditive.
	
	First of all, we assume that the images of singletons under $f$ are singletons. Since $f$ is union-subadditive and $f$ permutes the nonempty proper subsets of $S$, it forces that
	$f(\{x,y\})=f(\{x\})\cup f(\{y\})$ for all $x,y\in S$. Therefore, 
	$f=f_\sigma$
	where $\sigma$ is determined by $f(\{x\})=\{\sigma(x)\}$ for all $x\in S$. 
	
	Secondly, we assume none of the image of singletons under $f$ is a singleton. Then
	the images of singletons under $f^c$ must be singletons  because $f$ permutes the nonempty proper subsets of $S$. Since $f$ is intersection-subadditive, it follows that $f^c$ is union-subadditive by the De Morgan's Law.
	Hence, as in the first case, 
	$f^c=f_\sigma$ for some permutation $\sigma$ on $S$. Therefore, $f=f^c_\sigma$ in this case. 
	
Finally, we argue by contradiction that the case where
some of the images of singletons under $f$ are singletons and some are not is impossible. Without loss of generality, we may let $S=\{s_0,s_1,s_2\}$ and assume $f(\{s_0\})$ is a singleton while $f(\{s_1\})$ is not.
	
Case 1. $f(\{s_0\})\cap f(\{s_1\})\neq \emptyset$.\\
Then choose a permutation $\sigma\colon S\rightarrow S$ such that $\sigma[ f(\{s_0\})]=\{s_0\}$ and
$\sigma[ f(\{s_1\})]=\{s_0,s_1\}$. Let $g=f_\sigma\circ f$. Then $g(\{s_0\})=\{s_0\}$ and $g(\{s_1\})=\{s_0,s_1\}$. It is straightforward to show that $g$ is  union-subadditive, intersection-subadditive, and permutes the nonempty proper subsets of $S$.
	Then $g(\{s_0,s_1\})$ must be $\{s_1\}$ because $g(\{s_0,s_1\})\subseteq g(\{s_0\})\cup g(\{s_1\})$ (and $g$ permutes the nonempty proper subsets of $S$). It follows that $g(\{s_0,s_2\})=\{s_0,s_2\}$ because $g(\{s_0\})\subseteq g(\{s_0,s_1\})\cup g(\{s_0,s_2\})$. Hence, the remaining possible values for $g(\{s_1,s_2\})$ 
	are $\{s_2\}$ and $\{s_1, s_2\}$. However, this contradicts the fact that
	$g(\{s_1\})\subseteq g(\{s_0,s_1\})\cup g(\{s_1,s_2\})$.
	
	Case 2. $f(\{s_0\})\cap f(\{s_1\})=\emptyset$.
Then choose a permutation $\sigma\colon S\rightarrow S$ such that $\sigma[ f(\{s_0\})]=\{s_0\}$ and
$\sigma[ f(\{s_1\})]=\{s_1,s_2\}$. Let $g=f_\sigma\circ f$. Then $g(\{s_0\})=\{s_0\}$ and $g(\{s_1\})=\{s_1,s_2\}$.	
Similarly, $g$ is union-subadditive, intersection-subadditive, and permutes the nonempty proper subsets of $S$.
	
	Case 2.1. $g(\{s_2\})=\{s_1\}$ or $g(\{s_2\})=\{s_2\}$.\\
	Assume $g(\{s_2\})=\{s_1\}$. Then $g(\{s_0,s_2\})=\{s_0,s_1\}$ and $g(\{s_1,s_2\})=\{s_2\}$ due to union-subadditivity.
	Hence, $g(\{s_0,s_1\})$ takes the remaining value $\{s_0,s_2\}$. However, this contradicts $g(\{s_1\})\subseteq g(\{s_0,s_1\})\cup g(\{s_1,s_2\})$.
	The case $g(\{s_2\})=\{s_2\}$ is similar.
	
	Case 2.2. $g(\{s_2\})=\{s_0,s_1\}$ or $g(\{s_2\})=\{s_0,s_2\}$.\\
	Assume $g(\{s_2\})=\{s_0,s_1\}$. Then $g(\{s_0,s_2\})=\{s_1\}$ due to union-subadditivity.
	It follows that $g(\{s_0,s_1\})=\{s_0,s_2\}$ because $g(\{s_0\})\subseteq g(\{s_0,s_1\})\cup g(\{s_0,s_2\})$. Hence, $g(\{s_1,s_2\})$ takes the remaining value $\{s_2\}$. However, this contradicts the fact that $g(\{s_1\})\subseteq g(\{s_0,s_1\})\cup g(\{s_1,s_2\})$.
	The case $g(\{s_2\})=\{s_0,s_2\}$ is similar.
\end{proof}

The previous proof implied implicitly the following observations, straightforward proofs of which are omitted.

\begin{remark}\label{1606a}
Suppose $\sigma\colon S\rightarrow S$ and $f =\res_A$. Then
\begin{enumerate}
\item $f^c=\res_{A'}$, where $A'=\{\, (I_a,R_a, P_a)\mid a\in A\,\}$;
\item $f_{\sigma}\circ f= \res_{A'}$, where $A'=\{\, (R_a,I_a, \sigma[P_a])\mid a\in A\,\}$;
\item $f\circ f_\sigma = \res_{A'}$, where $A'=\{\, (\sigma^{-1}[R_a],\sigma^{-1}[I_a], P_a)\mid a\in A\,\}$.	
\end{enumerate}	
Therefore,  if $f\in \mathcal{M}(S)$, then
$f^c, f_{\sigma}\circ f\in \mathcal{M}(S)$ and if additionally $\sigma$ is one-to-one, then $f\circ f_\sigma\in \mathcal{M}(S)$.
\end{remark}

The following twelve rs functions over $S=\{s_0,s_1,s_2\}$ constitute the class $\mathcal{F}^P_U(S)$:
\begin{align*}
f_{\sigma_1}={}&[\mathbf{0} \;\,\mathbf{1} \;\,\mathbf{2} \;\,\mathbf{3} \;\,\mathbf{4}\;\, \mathbf{5}  \;\,\mathbf{6}\;\, \mathbf{0}  ],   & f_{\sigma_2}={}  &[\mathbf{0} \;\,\mathbf{2}\;\, \mathbf{1}\;\, \mathbf{3}\;\, \mathbf{4} \;\,\mathbf{6}\;\,  \mathbf{5}\;\, \mathbf{0}  ], & f_{\sigma_3}={}&[\mathbf{0}\;\, \mathbf{4}\;\, \mathbf{2}\;\, \mathbf{6}\;\, \mathbf{1}\;\, \mathbf{5}\;\,  \mathbf{3}\;\, \mathbf{0}  ], \\
 f_{\sigma_4}={} &[\mathbf{0}\;\, \mathbf{1}\;\, \mathbf{4}\;\, \mathbf{5}\;\, \mathbf{2}\;\, \mathbf{3} \;\, \mathbf{6}\;\, \mathbf{0}  ],&
f_{\sigma_5}={}  &[\mathbf{0}\;\, \mathbf{2}\;\, \mathbf{4}\;\, \mathbf{6}\;\, \mathbf{1}\;\, \mathbf{3}\;\,  \mathbf{5}\;\, \mathbf{0}  ],   & 
f_{\sigma_6}={}&[\mathbf{0}\;\, \mathbf{4}\;\, \mathbf{1}\;\, \mathbf{5}\;\, \mathbf{2}\;\, \mathbf{6}\;\,  \mathbf{3}\;\, \mathbf{0}  ],\\
f^c_{\sigma_1}={}&[\mathbf{0} \;\,\mathbf{6} \;\,\mathbf{5} \;\,\mathbf{4} \;\,\mathbf{3}\;\, \mathbf{2}  \;\,\mathbf{1}\;\, \mathbf{0}  ],   & 
f^c_{\sigma_2}={}&[\mathbf{0} \;\,\mathbf{5}\;\, \mathbf{6}\;\, \mathbf{4}\;\, \mathbf{3} \;\,\mathbf{1}\;\,  \mathbf{2}\;\, \mathbf{0}  ], &
f^c_{\sigma_3}={} &[\mathbf{0}\;\, \mathbf{3}\;\, \mathbf{5}\;\, \mathbf{1}\;\, \mathbf{6}\;\, \mathbf{2}\;\,  \mathbf{4}\;\, \mathbf{0}  ], \\
f^c_{\sigma_4}={}&[\mathbf{0}\;\, \mathbf{6}\;\, \mathbf{3}\;\, \mathbf{2}\;\, \mathbf{5}\;\, \mathbf{4} \;\, \mathbf{1}\;\, \mathbf{0}  ],&
f^c_{\sigma_5}={}&[\mathbf{0}\;\, \mathbf{5}\;\, \mathbf{3}\;\, \mathbf{1}\;\, \mathbf{6}\;\, \mathbf{4}\;\,  \mathbf{2}\;\, \mathbf{0}  ],   & 
f^c_{\sigma_6}={}&[\mathbf{0}\;\, \mathbf{3}\;\, \mathbf{6}\;\, \mathbf{2}\;\, \mathbf{5}\;\, \mathbf{1}\;\,  \mathbf{4}\;\, \mathbf{0}  ],
\end{align*}
where $\sigma_1$ is the identity permutation,
$\sigma_2= (s_0  \;\, s_1 )$, $\sigma_3= (s_0  \;\, s_2 )$,
$\sigma_4= (s_1  \;\, s_2 )$, $\sigma_5= (s_0  \;\,s_1 \;\,s_2 )$,
and $\sigma_6= (s_2 \;\,s_1 \;\, s_0 )$.
Due to our choice of representation, the row vector representing  the complement $f^c$ of an rs function $f$ is the mirror image of the row vector representing $f$. 

Going beyond the ternary alphabet, we found out computationally that the conclusion of Theorem~\ref{0106a} holds for the quarternary alphabet. Hence, we state the following theorem based on our computer simulation.

\begin{theorem}\label{0906a}
	For $\vert S\vert \in \{3,4\}$, the class $\mathcal{F}_U^P(S)$ is exactly the set of nondegenerate rs functions over $S$ that
	permute the nonempty proper subsets of $S$ and belong to $\mathcal{M}(S)$.
\end{theorem}

Notice that for a quarternary alphabet $S$, there are $14!=87,178,291,200$ rs functions over $S$ that
permute the nonempty proper subsets of $S$. It is surprising that only $48$ out of these can be 
specified by nondegenerate minimal reaction systems.\footnote{With our last refined simulation, the running time to obtain this computational result is pleasantly short, compared to which was needed for Theorem~\ref{0305a}.} 
Therefore, we ask whether the result of Theorem~\ref{0906a} holds for higher alphabets and if that is the case, a formal proof would be more desirable.

\begin{corollary}
Suppose $\vert S\vert\in \{3,4\}$. The set of nondegenerate rs functions over $S$ belonging to $\mathcal{M}(S)$ does not constitute a complete set of generating functions under composition for the set of all nondegenerate rs functions over $S$.
\end{corollary}

\begin{proof}
Suppose $f$ is any nondegenerate rs function over $S$  that permutes the nonempty proper subsets of $S$ but $f\notin \mathcal{F}^P_U(S)$. We claim that $f$ cannot be generated under composition by nondegenerate rs  functions belonging to $\mathcal{M}(S)$. We argue by contradiction.
Assume there is such a composition equaling $f$. Then each of the composing function must also permute the nonempty proper subsets of $S$ and thus is in $\mathcal{F}^P_U(S)$ by Theorem~\ref{0906a}. Since 
$\mathcal{F}^P_U(S)$ is closed under composition by Proposition~\ref{1806b},
it follows that $f\in \mathcal{F}^P_U(S)$, a contradiction.
\end{proof}

We end this section by mentioning further computational results for the ternary alphabet for inspiration. Similar attempt on the quarternary alphabet has to be aborted due to computational limitation. Out of $8^6=262144$ nondegenerate rs functions over a ternary alphabet $S$, there are  
$24389$ that belong to $\mathcal{M}(S)$ and exactly twelve among these, namely those in $\mathcal{F}^P_U(S)$, further permute the nonempty proper subsets of $S$.
Under composition, these $24389$ functions generate $257404$ nondegenerate rs functions over $S$, $98.2\%$ of all nondegenerate rs functions over $S$.

The following table shows the distribution of these
$257404$ rs functions according to $\vert \{ \,f(X) \mid \emptyset \neq X\subsetneq S\, \}	\vert$, which is the same as the \emph{genus} of $f$ restricted to  $2^S\backslash \{\emptyset, S\}$ and thus will be called here the \emph{$N\!$-genus of $f$}, with $N$ refers to nondegenerate.
\begin{table}[h]
\begin{center}
	\begin{tabular}{ | c | c | }
		\hline
$N$-genus of $f$ & Number of  $f$   \\ \hline
		1& 8  \\ \hline
		2 & 1736\\ \hline
		3 &  30240\\ \hline
		4 & 109200  \\ \hline
		5 &  100800    \\ \hline
		6 &   15420  \\ \hline
			\end{tabular}
\vspace{2mm}
\caption{Distribution of  rs functions over a ternary alphabet $S$ generated under composition from nondegenerate rs functions belonging to $\mathcal{M}(S)$} 
\end{center}
\end{table}

Our result in fact shows that every nondegenerate rs function over a ternary alphabet $S$ with $N\!$-genus less than the maximum six can be generated under composition from nondegenerate rs functions belonging to $\mathcal{M}(S)$. Among those with $N\!$-genus equal to six, there are $4740$ that cannot be thus generated, including $708$ ($= 6!-12$) that permute the nonempty proper subsets of $S$.


\section{Further Extension Allowing Degenerate Reaction Systems}

Now, we will study problems analogous to those encountered in the last section with the allowance of degenerate reaction systems. In this context, we are foremostly led to ask whether any permutation on the whole power set of $S$ that belongs to $\mathcal{M}(S)$ exists, apart from the identity permutation. The following example not only answers this, but shows that, unlike the ternary alphabet, the set of permutations on $2^S$ that belong to $\mathcal{M}(S)$ is not closed under composition. 

\begin{example}
Let $S=\{s_0,s_1,s_2\}$. Consider the rs functions $f$ and $g$ defined by
$f=(\mathbf{2}\;\,\mathbf{3})(\mathbf{6}\;\,\mathbf{7})$
and $g=(\mathbf{4}\;\,\mathbf{6})(\mathbf{5}\;\,\mathbf{7})$.
 Both $f$ and $g$ belong to $\mathcal{M}(S)$.
Consider the composition $h=f\circ g$  (from right to left). Then $h=(\mathbf{2}\;\,\mathbf{3})(\mathbf{4}\;\,\mathbf{7}\;\, \mathbf{5}\;\,\mathbf{6}     )$.
However, $h$ is not intersection-subadditive because
$S=h(\{s_2\})\nsubseteq h(\{s_0,s_2\})\cup h(\{s_1,s_2\})=\{s_1,s_2\}\cup \{s_2\}$. Therefore, $h$ does not belong to $\mathcal{M}(S)$.
\end{example}


There are $8^8=16777216$ rs functions over a ternary alphabet $S$.
Computationally, we found that there are 405224 rs functions over $S$ belonging to $\mathcal{M}(S)$. Among these, 408 are permutations.

\begin{theorem}\label{1804a}
Suppose $S$ is a ternary alphabet. There are exactly 408 rs functions over $S$ belonging to $\mathcal{M}(S)$ that are permutations on $2^S$, which include  $(\mathbf{2} \;\,\mathbf{7} \;\,	\mathbf{0} \;\,\mathbf{1} \;\, \mathbf{4} \;\,\mathbf{3}
\;\,\mathbf{6} \;\,\mathbf{5})$ and $(\mathbf{2} \;\,\mathbf{7})$. The latter two form a basis of $\Sym(2^S)$.
\end{theorem}

\begin{proof}
The first part is verified by our computer simulation. For the second part, it is a standard result in group theory that $(a_1 \;\,a_2 \;\,a_3 \;\,\dotsm \;\, a_n)$ and $(a_1 \:\,a_2)$ form a basis of $\Sym(A)$, where $A=\{a_1,a_2,a_3, \dotsc, a_n\}$.	
\end{proof}


The following classical result regarding composition of unary operations over a finite domain will be utilized.  

\begin{theorem}\cite{piccard1946bases, salomaa2003composition}\label{2206a}
	Suppose $A$ is a finite set of size $n\geq 3$. Then any three functions form a complete set of generating functions for the set of all functions from $A$ into $A$	
	if and only if two of them form a basis of $\Sym(A)$ and the genus of the third is $n-1$. 
\end{theorem}

\begin{corollary}
Suppose $S$ is a ternary alphabet.
There exist three fixed rs functions over $S$ belonging to $\mathcal{M}(S)$ that form a complete set of generating functions for the set of all rs functions over $S$.
\end{corollary}

\begin{proof}
The rs function $f_\sigma$ over $S$, where $\sigma$ is the identity permutation, has genus $\vert 2^S\vert -1$. Hence, $f_\sigma$ together with the permutations
$(\mathbf{2} \;\,\mathbf{7} \;\,	\mathbf{0} \;\,\mathbf{1} \;\, \mathbf{4} \;\,\mathbf{3}
\;\,\mathbf{6} \;\,\mathbf{5})$ and $(\mathbf{2} \;\,\mathbf{7})$ form a complete set of generating functions for the set of all rs functions over $S$ by Theorems~\ref{1804a} and \ref{2206a}.	
\end{proof}

When it comes to a quarternary alphabet $S$, the computational complexity increases dramatically: for example, there are $16^{16}$ rs functions over $S$, as opposed to only $8^8$ over a ternary alphabet. We are able to filter out rs functions over $S$ belonging to $\mathcal{M}(S)$ that are permutations by approximating them with partial functions. 

For $n\in \{0,1,2,\dotsc, 15\}$, let 
\begin{multline*}
A_n = \{ \, f\colon \{\mathbf{0},\mathbf{1},\mathbf{2},\dotsc,\mathbf{n}\}\rightarrow  \{\mathbf{0},\mathbf{1},\mathbf{2},\dotsc,\mathbf{15}\}   \mid f \text{ is one-to-one and }\\
f(X\cup Y) \cup f(X\cap Y)\subseteq f(X)\cup f(Y) \text{ for all } X,Y \in \{\mathbf{0},\mathbf{1},\mathbf{2},\dotsc,\mathbf{n}\} \,\}.
\end{multline*}

\begin{table}[h]
\begin{center}
	\begin{tabular}{ | c | c | c|c| }
		\hline
		n & Size of $A_n$    & n & Size of $A_n$ \\ \hline
		0 &   16      & 8 &  3463008    \\ \hline
		1 & 240  & 9 & 2835240\\ \hline
	2 & 1840 & 10 &  1337520\\ \hline
3 &  17776 & 11 & 855576\\ \hline
4 &  74952 & 12 & 170592 \\ \hline
5 &  223992 & 13 &  72216   \\ \hline
6 &   360540  & 14 & 42456\\ \hline
7 & 1110864  & 15 & 23424\\ \hline
	\end{tabular}
\vspace{2mm}
\caption{Size of $A_n$}
\end{center}
\end{table}

From the definition, the set $A_{15}$ consists of rs functions over $S$ that are permutations, union-subadditve and intersection-subadditive. Our approach significantly assists in reducing the complexity of our computation because for every integer $0\leq n\leq 14$ and $f\in A_{n+1}$, the restriction of $f$ to $\{\mathbf{0},\mathbf{1},\mathbf{2},\dotsc,\mathbf{n}\}$ belongs to $A_n$. Therefore, to obtain $A_{n+1}$, it suffices to find for every $f\in A_n$ every extension of $f$ by an additional value that fulfills the additional relevant clauses for union-subadditivity and intersection-subadditivity. 

We now state the following result based on our computer simulation. 

\begin{theorem}\label{0305a}
Suppose $S$ is a quarternary alphabet. There are exactly 23424 rs functions over $S$ belonging to $\mathcal{M}(S)$ that are permutations on $2^S$. Furthermore, these 23424 are all even permutations. 
\end{theorem}

The fact that no odd permutation on $2^S$ belongs to $\mathcal{M}(S)$ for a quarternary alphabet $S$ is surprising at first in view of Theorem~\ref{1804a}. Then it is plausible that the permutations on $2^S$ belonging to $\mathcal{M}(S)$ may generate under composition a proper subgroup of the alternating group on $2^S$. Interestingly, with some standard properties on symmetric groups and partial assistance from the computer,  we obtain the following result.

\begin{theorem}\label{1806c}
Suppose $S$ is a quarternary alphabet. Under composition, permutations on $2^S$ belonging to $\mathcal{M}(S)$ generate the alternating group on $2^S$.
\end{theorem}

\begin{proof}
Let $\mathcal{G}$ denote the subgroup of $\Sym(2^S)$
generated by permutations on $2^S$ belonging to $\mathcal{M}(S)$. By Theorem~\ref{0305a}, $\mathcal{G}$ is a subgroup of the alternating group on $2^S$.
Consider $p= (\mathbf{8} \;\,\mathbf{11} )(\mathbf{12} \;\,\mathbf{15})$ and $q= (\mathbf{4} \;\,\mathbf{13})(\mathbf{6} \;\,\mathbf{15})$. Both permutations belong to $\mathcal{M}(S)$. 
The composition $p\circ q \circ p \circ q$   (again from right to left) is equal to $(\mathbf{6} \;\,\mathbf{15} \;\,\mathbf{12})$
and thus the latter $3$-cycle belongs to $\mathcal{G}$.\footnote{It can be verified computationally that no 3-cycle belongs to $\mathcal{M}(S)$.}
Note that if $r\in \mathcal{G}$, then $r^{-1}\in \mathcal{G}$ because
$r^{-1}=r^{o(r)-1}$, where $o(r)$ is the order of the permutation $r$,
and thus  $(r(\mathbf{6}) \;\,  r( \mathbf{15}) \;\,  r(\mathbf{11}))
=      r\circ (\mathbf{6} \;\,\mathbf{15} \;\,\mathbf{12})\circ r^{-1} \in \mathcal{G}$.
Using computer, when we take a sufficiently large subset $R$ of $\mathcal{G}$ (for example, $R$ can be the collection of permutations that are obtained from a composition
of two permutations belonging to $\mathcal{M}(S)$),
 all 1120 ($= \perm{16}{3}/3$) permutations on $2^S$ that are 3-cycles are included in $\{\,(r(\mathbf{6}) \;\,  r( \mathbf{15}) \;\,  r(\mathbf{11})) \mid r \in R      \,\}$. It is a well-known fact that the set of those 3-cycles forms a complete set of generators for the alternating group on $2^S$. Therefore, $\mathcal{G}$ is the  alternating group on $2^S$.
\end{proof}

\begin{corollary}
Suppose $S$ is a quarternary alphabet. 
The set of rs functions over $S$ belonging to $\mathcal{M}(S)$ does not constitute a complete set of generating functions under composition for the set of all rs functions over $S$.
\end{corollary}

\begin{proof}
If a composition of rs functions is a permutation, then each of the composing function must be a permutation. 
Hence, by Theorem~\ref{1806c}, no odd permutation of $2^S$ can be generated by rs functions over $S$ belonging to $\mathcal{M}(S)$.
\end{proof}

We end this paper with an open problem. In view of 
Theorem~\ref{1806c}, it is intriguing to know 
which subgroup of $\Sym(2^S)$ do the permutations on $2^S$ belonging to $\mathcal{M}(S)$ generate for higher alphabets $S$. An immediate simpler question is whether only even permutations are generated. 

\section*{Acknowledgement}

The author is indebted to distinguished Professor Arto Salomaa for his valuable comments regarding this paper.



\end{document}